\documentclass[journal]{IEEEtran}

\usepackage{amsmath,amssymb,amsthm,mathtools,bm}
\usepackage{booktabs,array}
\usepackage{tikz}
\usetikzlibrary{bayesnet,positioning,arrows.meta,fit,backgrounds}
\usepackage[colorlinks=true,allcolors=blue]{hyperref}
\usepackage{cite}
 \usepackage{pgfplots}
    \pgfplotsset{compat=1.18}     \usepgfplotslibrary{groupplots}

\newtheorem{proposition}{Proposition}
\newtheorem{lemma}{Lemma}
\theoremstyle{definition}

\newtheorem{example}{Example}
\theoremstyle{remark}
\newtheorem{remark}{Remark}

 \DeclareMathOperator{\CVaR}{CVaR}
\newcommand{\Dinf}{\mathrm{D}_{\infty}}

\newcommand{\E}{\mathbb{E}}
\newcommand{\R}{\mathbb{R}}
\newcommand{\Sset}{\mathcal{S}}
\newcommand{\Oset}{\mathcal{O}}
\newcommand{\Aset}{\mathcal{A}}

\newcommand{\Dset}{\mathcal{D}}
\newcommand{\Pset}{\mathcal{P}}
\newcommand{\TV}{\mathrm{TV}}

\DeclareMathOperator*{\argmax}{arg\,max}

\begin{document}

\title{Decision Making Needs Uncertainty Quantification [Lecture Notes]}
\author{Osvaldo Simeone\thanks{O. Simeone is with the Institute for Intelligent Networked Systems, Northeastern University London, One Portsoken Street, London, E1 8PH, UK (email:  o.simeone@nulondon.ac.uk). The work of O. Simeone was supported by the European Research Council (ERC) under the European Union’s Horizon Europe Programme (grant agreement No. 101198347), by an Open Fellowship of the EPSRC (EP/W024101/1), and by the EPSRC project (EP/X011852/1). }}


\maketitle

\begin{abstract}
Many signal processing systems ultimately exist to {act}. Whenever the state variable that
determines the action to be taken by a decision maker, or agent, is uncertain, the way that uncertainty is
represented decides how
well the agent performs and how much its performance can be trusted. This
lecture note develops, from first principles and within a single 
decision-theoretic setting, the link between the {objective} and the knowledge of an agent and the form of uncertainty representation that is
sufficient to act optimally. To start, assuming a known environment distribution, we show that a risk-neutral agent needs the
posterior distribution over the state, whereas a risk-averse
agent can rely without loss of optimality on a {prediction set} and a worst-case decision rule. We then
turn to the  case in which the environment is unknown, and
identify three complementary approaches to address  the resulting epistemic
uncertainty: calibration of a fixed predictor, credal (ambiguity) sets with
distributionally robust optimization, and Bayesian inference over model
parameters. The common thread is that reliable decisions require an
uncertainty representation matched to the decision objective and to the knowledge profile of the agent,  together with
a guarantee that certifies the utility the agent
will actually obtain.
\end{abstract}

\section{Relevance}
Point predictions are rarely the end goal of a signal processing pipeline, which ultimately needs to produce an action or a decision. For example, based on a prediction, a detector triggers an alarm, a scheduler
assigns power or bandwidth, a controller selects a gain, a diagnostic system
recommends a test. This lecture note explains, at a level accessible to a
first-year graduate student, \emph{why} good decisions require uncertainty
quantification and \emph{which} representation of uncertainty is appropriate
for a given decision objective and state of knowledge. {The material connects three topics 
that are usually taught separately, namely Bayesian decision theory, calibration
of machine-learning predictors, and  robust
optimization, by viewing them as solutions to the same decision problem under different assumptions about the knowledge of the agent about the system state. {Overall, these notes demonstrate that a reliable decision agent needs tailored mechanisms for uncertainty quantification in order to provide predictable performance guarantees. It also connects and contrasts predictability with the separate requirements of consistency, robustness, and safety \cite{rabanser2026towards}.} }

\section{Prerequisites}
The reader is assumed to be comfortable with elementary probability
(conditional distributions, expectation, Bayes' rule), with the notion of an
optimization problem and its optimal value, and with basic estimation. No
prior exposure to conformal prediction, imprecise probability, or
distributionally robust optimization is required.  Familiarity with Bayesian inference at the
level of an introductory machine-learning course \cite{simeone2022ml} is
helpful for the final section.

\section{Problem Statement}
\label{sec:setting}
We consider an agent that acts once in a given environment. The problem is
described by the following random variables:
\begin{itemize}
\item a \emph{state} $s\in \Sset$ with distribution $s\sim p(s)$;
\item an \emph{observation} $o\in \Oset$ with conditional distribution $p(o| s)$;
\item an \emph{action} $a\in \Aset$ drawn from a policy $\pi(a| o)$ as $a \sim \pi(a| o)$;
\item and a scalar \emph{utility}, or reward, $u(s,a)$, depending on state $s$ and action $a$.
\end{itemize}
The joint distribution of the random variables  $(s,o,a)$, taking values in sets $\mathcal{S}$, $\mathcal{O}$, and $\mathcal{A}$, respectively,  factorizes as
\begin{equation}
\label{eq:joint}
p(s,o,a)=\underbrace{p(s)\,p(o| s)\,}_{\text{environment }p(s,o)}\underbrace{\pi(a| o)}_{\text{policy}},
\end{equation}
with utility $u(s,a)$. A graphical representation of the joint distribution, known as  Bayesian network,   is shown in
Fig.~\ref{fig:bn}. The agent controls the policy $\pi(a| o)$, while the
\emph{environment} is described by the joint distribution
$p(s,o)=p(s)\,p(o| s)$, which dictates the state distribution $p(s)$ and the
observation model $p(o|s)$. We assume throughout that the agent cannot modify the
environment $p(s,o)$ and has access only to the observation $o$, and not to
the underlying state $s$, when selecting an action $a\sim \pi(a|o)$.

The central question of this lecture note is: \emph{what must the agent know
about the uncertain state $s$, and in what form, in order to act optimally?} As
we will see, the answer depends both on the agent's attitude toward risk and
on whether the environment is known.

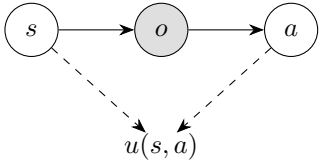
\begin{figure}[t]
\centering
\begin{tikzpicture}[
  >={Stealth[length=2mm]},
  node distance=14mm,
  every node/.style={inner sep=2pt}]
  \node[latent] (s) {$s$};
  \node[obs, right=of s] (o) {$o$};
  \node[latent, right=of o] (a) {$a$};
  \node[const, below=10mm of o] (u) {$u(s,a)$};
  \edge {s} {o};
  \edge {o} {a};
  \draw[->,dashed] (s) -- (u);
  \draw[->,dashed] (a) -- (u);
\end{tikzpicture}
\caption{Bayesian network for the decision problem under study in these notes: The state $s$ is hidden, and the only information available on it is through the observation $o$; the agent's policy maps the observation $o$ to an action
$a$; and the utility $u(s,a)$ 
depends on the state $s$ and the action $a$.}
\label{fig:bn}
\end{figure}

\begin{remark}[Goal-oriented decision making]
The setting extends to a goal-oriented decision maker by introducing a
\emph{goal} variable $g\sim p(g)$. This only changes the policy to
$a \sim \pi(a| g,o)$ and the utility to $u( g,s,a)$. Accordingly, all the
results below extend directly by conditioning on the goal.
\end{remark}

{ 
\begin{remark}[Action space] The action space $\mathcal{A}$ may include special actions such as abstention and ``safe'' actions. With abstention,  the agent refuses to take an action in the environment, possibly deferring to a higher authority such as a human in the loop or a cloud-based large model. An abstention  action may be assigned a negative utility when judging the performance of the agent; or, alternatively, the utility accrued by the final, deferred, action, when evaluating the end-to-end performance of the system. ``Safe'' action may come with utility guarantees that hold under any state compatible with the given observation, offering a default choice that does not achieve peak performance but yields deterministically predictable outcomes.

\end{remark}}
\section{Solution}
We build the answer in stages. Sections~\ref{sec:known} and \ref{sec:ra}
treat a \emph{known} environment and show how the decision objective
(risk-neutral vs.\ risk-averse) dictates the form of uncertainty
representation. Section~\ref{sec:unknown} then confronts an \emph{unknown}
environment and introduces three complementary tools to address the resulting
epistemic uncertainty. The five resulting settings are summarized in
Fig.~\ref{fig:racinterface}.

Throughout the text, we denote the distribution of a  random variable $x\in \mathcal{X}$ with support given by set $\mathcal{X}$ as $p(x)$, not differentiating between the value $p(x)$ for a given numerical realization $x$ and the entire function $\{p(x)\}_{x\in \mathcal{X}}$.  We also do not differentiate between random variable $x$ and its realization, which is also denoted as $x$. We use $\E_{x\sim p(x)}[\cdot]$ and $\Pr_{x\sim p(x)}[\cdot]$ for expectation and probability evaluated with respect to distribution $p(x)$.

\subsection{Known Environment, Risk-Neutral Agent: The Posterior as Interface}
\label{sec:known}
A \emph{risk-neutral} agent optimizes the average utility, solving the problem
\begin{equation}\label{eq:opt1}
\max_{\pi(a| o)}\ \E_{(s,o,a)\sim p(s,o)\pi(a| o)}\!\left[u(s,a)\right],
\end{equation} where the utility $u(s,a)$ is averaged over both the environment distribution $p(s,o)$ and the policy $\pi(a|o)$. 
By the law of iterated expectations, the optimization (\ref{eq:opt1}) decouples across
observations, yielding the equivalent form \begin{equation}\label{eq:inner}
\E_{o\sim p(o)}\!\left[\max_{\pi(a| o)} \E_{a\sim \pi(a|o)}[U(a| o)]\right],
\end{equation} where
\begin{equation}\label{eq:uoa}
U(a| o)=\E_{s\sim p(s| o)}\!\left[u(s,a)\right]
\end{equation}
is the average utility of action $a$ conditioned on the environment producing
observation $o$. The conditional state distribution used in the averaged utility (\ref{eq:uoa}) follows from Bayes' rule as
\begin{equation}
\label{eq:posterior}
p(s| o)\propto p(s)\,p(o| s).
\end{equation}  { The distribution $p(s|o)$ captures the inherent, \emph{aleatoric}, uncertainty affecting the state $s$ for the given observation $o$ \cite{simeone2022ml}. }

By the equivalent form (\ref{eq:inner}), the agent can optimize its action separately for each observation $o$ based on the conditional state distribution $p(s|o)$. Furthermore, because the objective in the inner optimization in \eqref{eq:inner} is linear in the policy, the problem 
is solved by a deterministic policy choosing the action
\begin{equation}\label{eq:optaction}
a^*(o)=\argmax_{a \in \Aset} U(a| o).
\end{equation} An example of utility $u(s,a)$, posterior $p(s|o)$, and average utility $U(a|o)$ is shown in Fig. \ref{fig:ex}.

The key result of this section is that the posterior distribution $p(s| o)$ is a {sufficient statistic} of the
observation $o$ for the decision problem. An optimal agent need retain
nothing about $o$ or the environment beyond $p(s| o)$. Equivalently, as summarized in  Fig.~\ref{fig:racinterface}(a), the
map $o\mapsto p(s| o)$ is the optimal \emph{interface} between the agent
and its environment.

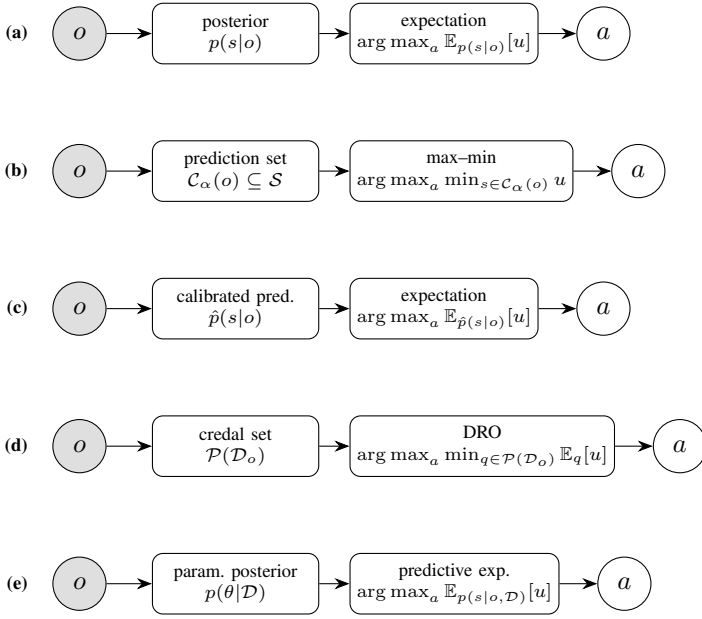
\begin{figure}[t]
\centering
\begin{tikzpicture}[
  >={Stealth[length=2mm]},
  box/.style={draw,rounded corners,inner sep=2.5pt,align=center,font=\scriptsize,
              minimum height=8mm,minimum width=22mm},
  panel/.style={font=\scriptsize\bfseries, anchor=east},
  every node/.style={inner sep=2pt}]
  \def\rowsep{11mm}
  \node[obs] (o1) at (0,0) {$o$};
  \node[box, right=6mm of o1] (i1) {posterior\\$p(s| o)$};
  \node[box, right=4mm of i1] (r1) {expectation\\$\argmax_a\E_{p(s| o)}[u]$};
  \node[latent, right=5mm of r1] (a1) {$a$};
  \draw[->] (o1) -- (i1); \draw[->] (i1) -- (r1); \draw[->] (r1) -- (a1);
  \node[panel] at ([xshift=-2.5mm]o1.west) {(a)};
  \node[obs, below=\rowsep of o1] (o2) {$o$};
  \node[box, right=6mm of o2] (i2) {prediction set\\$\mathcal{C}_{\alpha}(o)\subseteq\Sset$};
  \node[box, right=4mm of i2] (r2) {max--min\\$\argmax_a\min_{s\in \mathcal{C}_{\alpha}(o)}u$};
  \node[latent, right=5mm of r2] (a2) {$a$};
  \draw[->] (o2) -- (i2); \draw[->] (i2) -- (r2); \draw[->] (r2) -- (a2);
  \node[panel] at ([xshift=-2.5mm]o2.west) {(b)};
  \node[obs, below=\rowsep of o2] (o3) {$o$};
  \node[box, right=6mm of o3] (i3) {calibrated pred.\\$\hat{p}(s| o)$};
  \node[box, right=4mm of i3] (r3) {expectation\\$\argmax_a\E_{\hat{p}(s| o)}[u]$};
  \node[latent, right=5mm of r3] (a3) {$a$};
  \draw[->] (o3) -- (i3); \draw[->] (i3) -- (r3); \draw[->] (r3) -- (a3);
  \node[panel] at ([xshift=-2.5mm]o3.west) {(c)};
  \node[obs, below=\rowsep of o3] (o4) {$o$};
  \node[box, right=6mm of o4] (i4) {credal set\\$\Pset(\Dset_o)$};
  \node[box, right=4mm of i4] (r4) {DRO\\$\argmax_a\min_{q\in\Pset(\Dset_o)}\E_{q}[u]$};
  \node[latent, right=5mm of r4] (a4) {$a$};
  \draw[->] (o4) -- (i4); \draw[->] (i4) -- (r4); \draw[->] (r4) -- (a4);
  \node[panel] at ([xshift=-2.5mm]o4.west) {(d)};
  \node[obs, below=\rowsep of o4] (o5) {$o$};
  \node[box, right=6mm of o5] (i5) {param.\ posterior\\$p(\theta|\Dset)$};
  \node[box, right=4mm of i5] (r5) {predictive exp.\\$\argmax_a\E_{p(s| o,\Dset)}[u]$};
  \node[latent, right=5mm of r5] (a5) {$a$};
  \draw[->] (o5) -- (i5); \draw[->] (i5) -- (r5); \draw[->] (r5) -- (a5);
  \node[panel] at ([xshift=-2.5mm]o5.west) {(e)};
\end{tikzpicture}
\caption{The optimal {interface} between the agent and its environment
across the five settings studied in the text: Each
row maps the observation $o$ (with the dataset $\Dset$ or dataset $\mathcal{D}_o$ where available) to a
{sufficient statistic} and then to a {decision rule} selecting the
action $a$. \textbf{(a)} \emph{Risk-neutral agent in a known environment} (Sec.~\ref{sec:known}): The
statistic is the exact state posterior $p(s| o)$ and the rule is 
maximization of the expectation as in  \eqref{eq:optaction}. \textbf{(b)} \emph{Risk-averse agent in a known environment} 
(Sec.~\ref{sec:ra}): The statistic is a prediction set
$\mathcal{C}_{\alpha}(o)$ with guaranteed miscoverage probability $\alpha$, and the rule is the max--min policy
\eqref{eq:maxmin}. \textbf{(c)} \emph{Unknown environment and fixed predictor} (Sec.~\ref{sec:fixed}):
The statistic is the fixed predictor $\hat{p}(s| o)$ and 
maximization of the expectation in  \eqref{eq:optaction1} is optimal on average (among all policies dependent on the fixed predictor)  provided the
predictor is calibrated. \textbf{(d)} \emph{Unknown environment and data-driven non-parametric predictor}
(Sec.~\ref{sec:npar}): The statistic is the credal set $\Pset(\Dset_o)$ of
plausible posteriors given the observation-specific dataset $\mathcal{D}_o$ and the action follows the distributionally robust rule
\eqref{eq:mmaction}. \textbf{(e)} \emph{Unknown environment and data-driven parametric predictor} 
(Sec.~\ref{sec:par}): The statistic is the parameter posterior
$p(\theta|\Dset)$, which induces the predictive distribution $p(s| o,\Dset)$, and the action maximizes the corresponding average utility as in 
\eqref{eq:bayesact}.}
\label{fig:racinterface}
\end{figure}

\subsection{Risk-Averse Agent: Prediction Sets as Interface}
\label{sec:ra}
Instead of the mean, an agent may value the utility it can {guarantee
with high probability} over the unknown state $s\sim p(s| o)$. Fix a risk
level $\alpha\in(0,1)$. For an observation $o$ and action $a$, the
\emph{value-at-risk} at level $\alpha$ is
\begin{equation}
\label{eq:var}
V_\alpha(a| o)=\sup\Bigl\{\nu\in\R:\
\Pr_{s\sim p(s| o)}\!\left[\,u(s,a)\ge\nu\,\right]\ge 1-\alpha\Bigr\}.
\end{equation}
This is the largest utility attained with probability at least $1-\alpha$
under the posterior, or, equivalently, the lower $\alpha$-quantile of the utility
$u(s,a)$. The risk-averse agent maximizes it by addressing the problem 
\begin{equation}\label{eq:optactionalpha}
a_\alpha^*(o)=\argmax_{a \in \Aset} V_\alpha(a| o).
\end{equation}
As $\alpha\to0$ the objective concentrates on the worst outcomes, trading
expected utility for protection against low-utility states.

As shown next, this quantile objective can be recast as a worst-case
optimization over a {prediction set}, motivating a different interface
between decision maker and environment -- from a probability distribution to
a set. To elaborate, let a \emph{level-$\alpha$ prediction set} be a map
$o\mapsto \mathcal{C}_{\alpha}(o)\subseteq\Sset$ satisfying the coverage
guarantee
\begin{equation}
\label{eq:coverage}
\Pr_{s\sim p(s| o)}\!\left[\,s\in \mathcal{C}_{\alpha}(o)\,\right]\ge 1-\alpha.
\end{equation} By (\ref{eq:coverage}), the set $\mathcal{C}_{\alpha}(o)$ is guaranteed to contain the true state $s$ with probability no smaller than $1-\alpha$ given the observation $o$. An example of a level-$\alpha$ prediction set is shown in Fig. \ref{fig:ex}.

\begin{figure}[t]
  \centering
  \includegraphics[width=\columnwidth]{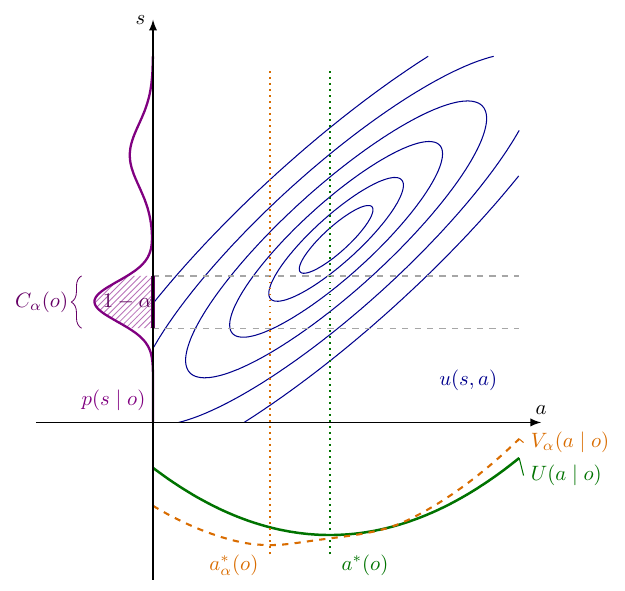}  
  \caption{Risk-neutral versus risk-averse decision making for a fixed
  observation~$o$. The  blue curves are level sets of the utility
  $u(s,a)$ over the state $s$ (vertical axis) and the action $a$
  (horizontal axis). The posterior $p(s|o)$ is shown against the
  state axis (left); the level-$\alpha$ prediction set $C_\alpha(o)$
  collects states carrying posterior mass $1-\alpha$ (hatched). 
  The average utility $U(a| o)$
  (solid) and the value-at-risk $V_\alpha(a| o)$ (dashed) yield the optimal action $a^*(o)$ in~\eqref{eq:optaction} and 
  $a^*_\alpha(o)$ in \eqref{eq:optactionalpha}, respectively. }
  \label{fig:ex}
\end{figure}

Given such a set, as also illustrated in Fig. \ref{fig:ex}, the \emph{max--min} policy selects the action maximizing
the worst-case utility within the set,
\begin{equation}
\label{eq:maxmin}
a_{\mathrm{mm}}(\mathcal{C}_{\alpha}(o))=\argmax_{a\in\Aset}\ \min_{s\in \mathcal{C}_{\alpha}(o)}u(s,a),
\end{equation}
with corresponding worst-case value
\begin{equation}
\label{eq:umm}
U_{\mathrm{mm}}(\mathcal{C}_{\alpha}(o))=\max_{a\in\Aset}\ \min_{s\in \mathcal{C}_{\alpha}(o)}u(s,a).
\end{equation}
As shown next, this worst-case value is an achievable value-at-risk for the
max--min action.

\begin{lemma}[Coverage yields a value-at-risk guarantee]
\label{lem:guarantee}
For any prediction set $\mathcal{C}_{\alpha}(o)$ satisfying the coverage
condition \eqref{eq:coverage}, the associated max--min action
$a_{\mathrm{mm}}(\mathcal{C}_{\alpha}(o))$ achieves a value-at-risk no smaller
than $U_{\mathrm{mm}}(\mathcal{C}_{\alpha}(o))$, i.e.,
\begin{equation}
V_\alpha\!\left(a_{\mathrm{mm}}(\mathcal{C}_{\alpha}(o))| o\right)\ge U_{\mathrm{mm}}(\mathcal{C}_{\alpha}(o)).
\end{equation}
\end{lemma}
\begin{proof}
The claim is equivalent to the inequality
\begin{equation}
\label{eq:mmguarantee}
\Pr_{s\sim p(s| o)}\!\left[\,u\!\left(s,a_{\mathrm{mm}}(\mathcal{C}_{\alpha}(o))\right)\ge U_{\mathrm{mm}}(\mathcal{C}_{\alpha}(o))\,\right]\ge 1-\alpha.
\end{equation}
Write $a_{\mathrm{mm}}=a_{\mathrm{mm}}(\mathcal{C}_{\alpha}(o))$ to simplify the notation. Every
state $s\in \mathcal{C}_{\alpha}(o)$ satisfies the inequality 
\begin{equation} u(s,a_{\mathrm{mm}})\ge\min_{s'\in \mathcal{C}_{\alpha}(o)}u(s',a_{\mathrm{mm}})=U_{\mathrm{mm}}(\mathcal{C}_{\alpha}(o)),\end{equation}
so the event $\{s\in \mathcal{C}_{\alpha}(o)\}$ is contained in the event 
$\{u(s,a_{\mathrm{mm}})\ge U_{\mathrm{mm}}(\mathcal{C}_{\alpha}(o))\}$. Hence, by
\eqref{eq:coverage}, we have the inequality 
$\Pr_{s\sim p(s|o)}[u(s,a_{\mathrm{mm}})\ge U_{\mathrm{mm}}(\mathcal{C}_{\alpha}(o))]\ge\Pr_{s\sim p(s|o)}[s\in \mathcal{C}_{\alpha}(o)]\ge 1-\alpha$. This coincides with \eqref{eq:mmguarantee}, thus concluding the proof. 
\end{proof}

Lemma~\ref{lem:guarantee} shows that {any} prediction set with a
coverage guarantee yields, through the max--min rule, an action with a
certified value-at-risk. The next proposition shows that the value-at-risk
optimal action \eqref{eq:optactionalpha} is recovered by choosing the
coverage set optimally, so that a max--min decision based on a suitably
designed prediction set incurs no loss of optimality. 

\begin{proposition}[Prediction sets are the optimal interface for risk-averse decisions]
\label{prop:ra}
The optimal risk-averse action \eqref{eq:optactionalpha} is the max--min
action of an optimally chosen level-$\alpha$ prediction set, i.e., 
\begin{equation}
\label{eq:raset}
a_\alpha^*(o)=a_{\mathrm{mm}}\!\left(C_\alpha^*(o)\right),
\end{equation}
where the optimized prediction set is given by
\begin{equation}
\label{eq:optset}
C_\alpha^*(o)\in\!\!\argmax_{\substack{\mathcal{C}\subseteq\Sset:\\ \Pr_{s\sim p(s| o)}[s\in \mathcal{C}]\ge 1-\alpha}}\!\!U_{\mathrm{mm}}(\mathcal{C}).
\end{equation}
\end{proposition}
\begin{proof}
We first establish, for a fixed action $a$, the identity
\begin{equation}
\label{eq:varset}
V_\alpha(a| o)=\!\!\max_{\substack{\mathcal{C}\subseteq\Sset:\\ \Pr_{s\sim p(s| o)}[s\in \mathcal{C}]\ge 1-\alpha}}\!\!\min_{s\in \mathcal{C}}u(s,a),
\end{equation}
which states that the value-at-risk equals the best worst-case utility achievable over a
covering set. Given this result, optimizing the value-at-risk over the action $a$ yields the max-min action $a_{\mathrm{mm}}\!\left(C_\alpha^*(o)\right)$, as claimed in the proposition. 

For the ``$\ge$'' direction in \eqref{eq:varset}, take any
feasible set $\mathcal{C}$, satisfying the coverage condition (\ref{eq:coverage}), and let $m(\mathcal C)=\min_{s\in \mathcal{C}}u(s,a)$. Then, we have the inclusion relationship 
$\{s\in \mathcal{C}\}\subseteq\{u(s,a)\ge m(\mathcal C)\}$, which implies the inequality 
$\Pr_{s\sim p(s| o)}[u(s,a)\ge m(\mathcal C)]\ge\Pr_{s\sim p(s| o)}[s\in \mathcal{C}]\ge1-\alpha$. In turn,  by the
definition \eqref{eq:var} of the value-at-risk, this gives the inequality 
$m(\mathcal{C})\le V_\alpha(a| o)$, which holds also when optimizing over the set. This proves the desired  $\ge$ inequality in \eqref{eq:varset}.

For the reverse direction, consider the set
$\mathcal{C}_0=\{s:u(s,a)\ge V_\alpha(a|o)\}$. By definition of value-at-risk, the set $\mathcal{C}_0$ is  feasible for problem \eqref{eq:varset}, and
satisfies the inequality $\min_{s\in C_0}u(s,a)\ge V_\alpha(a|o)$, proving the desired $\le$ inequality in \eqref{eq:varset}.

\end{proof}

The results developed in this section show that  a risk-{averse} agent can operate without loss of optimality on the basis of a {prediction set}
$\mathcal{C}_{\alpha}(o)\subseteq\Sset$ with coverage guarantees, acting by targeting the worst case over that set \eqref{eq:maxmin}. This is illustrated in Fig.~\ref{fig:racinterface}(b). In this regard, by Lemma~\ref{lem:guarantee}, any set satisfying the coverage condition \eqref{eq:coverage} yields a
max--min action with the value-at-risk guarantee \eqref{eq:mmguarantee}. Moreover, an optimized  set
-- namely $C_\alpha^*(o)$ in \eqref{eq:optset} -- yields the largest value-at-risk certificate, 
recovering the value-at-risk optimal action $a_\alpha^*(o)$.

So far coverage \eqref{eq:coverage} has been imposed \emph{conditionally},
i.e., for each observation $o$. One may instead require coverage \emph{on
average} over observations,
\begin{equation}
\label{eq:margcoverage}
\Pr_{(s,o)\sim p(s,o)}\!\left[\,s\in \mathcal{C}_{\alpha}(o)\,\right]\ge 1-\alpha,
\end{equation}
which is weaker than \eqref{eq:coverage}. Prediction sets with the marginal
guarantee \eqref{eq:margcoverage} can be constructed by \emph{conformal
prediction} using only data, without knowledge of the environment 
\cite{vovk2005algorithmic,angelopoulos2023gentle}. Reference
\cite{kiyani2025conformal} shows that results analogous to
Lemma~\ref{lem:guarantee} and Proposition~\ref{prop:ra} hold when the
value-at-risk objective \eqref{eq:var} is replaced by a suitable average,
placing conformal prediction on a decision-theoretic footing. The next sections address the setting with an unknown environment.

\subsection{Unknown Environment: Three Routes to Uncertainty Quantification}
\label{sec:unknown}
Suppose now that the environment distributions $p(s)$ and $p(o| s)$ are unknown. We treat this setting in three ways,
each highlighting a different approach to uncertainty quantification. For a preview of the main results, 
Fig.~\ref{fig:racinterface}(c)--(e) summarizes the resulting interfaces.

\noindent $\bullet$  \emph{Fixed predictor.} A predictor $\hat{p}(s| o)$ is supplied to the agent.
Under what conditions can the agent act on it, and
trust that the resulting performance is as good as it can be given that the policy can be based solely on the predictor? The answer motivates \emph{calibration} as a key property of predictors 
(Sec.~\ref{sec:fixed}).

\noindent $\bullet$  \emph{Data-driven non-parametric predictor.} No predictor is given,
but a dataset
\begin{equation}
\label{eq:dataset}
\Dset=\{(o_i,s_i)\}_{i=1}^{n}\ \stackrel{\text{i.i.d.}}{\sim}\ p(s,o)=p(s)\,p(o| s)
\end{equation}
of $n$ observation--state pairs is available. Because the dataset size $n$ is finite, plugging
a naive estimate into the known-environment solution is generally suboptimal, and 
the agent must account for the epistemic uncertainty caused by the availability of  limited data. { A
non-parametric treatment of this uncertainty requires making no assumptions about the environment, reducing bias \cite{simeone2022ml} but typically increasing data requirements, i.e., the dataset size $n$. With non-parametric estimators, an optimized interface with the decision maker will be seen to encompass  \emph{credal sets} and
distributionally robust optimization (Sec.~\ref{sec:npar}) in order to account for epistemic uncertainty.}

\noindent $\bullet$  \emph{Data-driven parametric predictor.} {Under parametric  
assumptions on the environment, the agent can generally decrease the sample size $n$, but at the cost of potentially introducing an irreducible bias in the training of a predictor \cite{simeone2022ml}. In this case,  epistemic uncertainty can be optimally accounted for by \emph{Bayesian
inference} over model parameters (Sec.~\ref{sec:par}).}

\subsection{ Fixed Predictors and Calibration}
\label{sec:fixed}
Assume a fixed predictor $\hat p(s| o)$, generally not equal to the true
posterior $p(s| o)$. Treating the predictor $\hat{p}(s| o)$ as if it were the true posterior, the
agent takes, in place of \eqref{eq:optaction}, the action
\begin{equation}\label{eq:optaction1}
a^*(\hat{p}(s| o))=\argmax_{a \in \Aset} \hat U(a| o),
\end{equation}
where the estimated utility under the predictor is
\begin{equation}
\label{eq:uhat}
\hat U(a| o)= \E_{s \sim \hat{p}(s| o)}[u(s,a)].
\end{equation}
When does  action (\ref{eq:optaction1}) guarantee some notion of optimality? When can the utility estimate be taken as a reliable estimate of the action's performance? As seen in this section, this is the case  when the predictor is
\emph{calibrated}.

A predictor is \emph{distribution calibrated} if, among all observations $o$
that yield any given predictive distribution $q(s)$, i.e.\
$\hat{p}(s| o)=q(s)$, the true fraction of states equal to $s$ equals the
predicted value $q(s)$. Formally, with
\begin{equation}
\label{eq:Oq}
\Oset_{q(s)}=\{o\in \Oset:\hat p(s'| o)=q(s')\ \text{for all }s'\in\Sset\}
\end{equation}
denoting the set of observations mapped to the predictive distribution $q(s)$, a
distribution-calibrated predictor satisfies
\begin{equation}
\label{eq:distcal}
p(s' | o\in\Oset_{q(s)})=q(s')\qquad \text{for all } s'\in\Sset
\end{equation} and for all predictive distributions $q(s)$ with non-zero probability. 
In words, a distribution-calibrated model provides predictive probabilities
that reflect the true conditional state frequencies.

\begin{example}[Calibration is not accuracy]
\label{ex:cal}
Two very different predictors are both distribution calibrated. As expected, the true posterior $\hat p(s| o)=p(s| o)$ is distribution calibrated, since each observation $o$ in set $\Oset_{p(s|o)}$ has true conditional distribution $p(s|o)$, implying the condition (\ref{eq:distcal}). However, a completely uninformative predictor such as the marginal $\hat p(s| o)=p(s)$, which ignores the
observation $o$, is also distribution calibrated. In fact, here every observation maps to the same predictive
distribution $q(s)=p(s)$, so the set  $\Oset_{q(s)}$ coincides with the entire observation domain $\mathcal O$, and 
\begin{equation}
p(s'| o\in\Oset_{q(s)})=p(s'| o\in\Oset)=p(s')
\end{equation} for all $s' \in \mathcal{S}$,   
recovering the condition \eqref{eq:distcal}. Both predictors are calibrated, yet only the posterior is
informative, showing that calibration is distinct from informativeness and accuracy.
\end{example}

{
\begin{example}[Visualizing calibration and accuracy]
The distinction between accuracy and calibration for a predictor $\hat p(s|o)$ can be conveniently illustrated visually by focusing on a setting with a binary state $s\in\mathcal{S}=\{0,1\}$. Denote as $\hat p \in [0,1]$ the predictive probability $\hat p(s=1|o)$, which determines the  entire predictive distribution $\hat p_o$ in this example. Accuracy and calibration are properties of the joint distribution $p(\hat p,s)$ of the prediction $\hat p$ and of the true state $s$.

Accuracy is high when the conditional distributions $p(\hat p|s=0)$ and $p(\hat p|s=1)$ are skewed, respectively,  towards $s=0$ and $s=1$, having a small overlap.  Examples of such distributions can be found in panels (a) and (c) of Fig. \ref{fig:cal-vs-acc}, while panels (b) and (d) show distributions that are significantly overlapped, corresponding to a lower accuracy. So,  assessing accuracy requires a \emph{horizontal} reading of the plot, as it is dictated by how far apart the two densities are located  along the $\hat p$ axis, i.e., by how little they overlap. 

In the setting of this example, the calibration condition in \eqref{eq:distcal} can be expressed as\begin{equation}
p(s=1|\hat p)=\hat p, 
\label{eq:distcal-binary}
\end{equation} which must hold for all predictive probabilities $\hat p \in[0,1]$ with non-zero probability. This requires that the fraction of outcomes with $s=1$ when the predictive probability is $\hat{p}$ is actually equal to $\hat{p}$. Using Bayes theorem and assuming for simplicity that the state is uniformly distributed, this condition can be rewritten as 
\begin{equation}
\frac{\,p(\hat p|s=1)}{\,p(\hat p|s=1)+\,p(\hat p|s=0)}=\hat{p}.
\label{eq:posterior-ratio}
\end{equation} Accordingly, in the plots of Fig. \ref{fig:cal-vs-acc}, calibration requires a \emph{vertical}, pointwise reading. In order for the predictor to be perfectly calibrated, at each value  $\hat p$, the fraction of the total height carried by the curve $p(\hat p|s=1)$ must equal $\hat p$ itself. 

The panels in Fig. \ref{fig:cal-vs-acc} show the probability $p(s=1|\hat{p})$  along with the prediction $\hat p$  (for values of $\hat{p}$ with positive probability).  It is observed that the distributions in panel (a) correspond to a setting with a calibrated and accurate predictor; panel (b) to a calibrated but inaccurate predictor, reproducing the predictor always returning the state marginal $p(s)=0.5$ discussed in Example 1; panel (c) to a miscalibrated and accurate predictor, which over-estimates the probability of event $s=1$ for $\hat p \leq 0.5$ and under-estimates it otherwise; while panel (d) represents a miscalibrated and
inaccurate predictor, which over-estimates and under-estimates the probability of
the event $s=1$ in alternating intervals of $\hat p$.
\end{example}

\begin{figure}[t]
\centering
\begin{tikzpicture}
\begin{groupplot}[
    group style={group size=2 by 2, horizontal sep=1.0cm, vertical sep=1.0cm,
                 xlabels at=edge bottom},
    width=3.35cm, height=2.75cm, scale only axis,
    xmin=0, xmax=1, ymin=0, ymax=2.2,
    xtick={0,0.5,1},
    ytick={0,0.5,1,2},
    tick label style={font=\tiny},
    label style={font=\scriptsize},
    title style={font=\scriptsize, yshift=-2pt},
    xlabel={$\hat p$},
    samples=201, domain=0:1,
    every axis plot/.append style={line width=0.8pt},
    clip=true,
]

\nextgroupplot[title={(a) calibrated, accurate},
    legend to name=leg:calvsacc,
    legend columns=2, legend cell align=left,
    legend style={draw=black, line width=0.4pt, inner sep=4pt, font=\scriptsize,
                  column sep=4pt,
                  /tikz/every even column/.append style={column sep=14pt}}]
\addplot[blue!70!black] {2*x};
\addlegendentry{$p(\hat p|s{=}1)$}
\addplot[red!75!black, dashed] {2*(1-x)};
\addlegendentry{$p(\hat p|s{=}0)$}
\addplot[black!55, line width=0.6pt] {x};
\addlegendentry{$p(s=1|\hat p)$}
\addplot[black!30, densely dotted, line width=0.6pt] {x};
\addlegendentry{$\hat p$}

\nextgroupplot[title={(b) calibrated, inaccurate}]
\addplot[blue!70!black, ->] coordinates {(0.5,0) (0.5,1)};
\addplot[red!75!black, dashed, ->] coordinates {(0.5,0) (0.5,1)};
\addplot[black!55, only marks, mark=*, mark size=1.1pt] coordinates {(0.5,0.5)};
\addplot[black!30, only marks, mark=o, mark size=2.2pt] coordinates {(0.5,0.5)};

\nextgroupplot[title={(c) miscalibrated, accurate}]
\addplot[blue!70!black] coordinates {(0,0) (0.5,0) (0.5,2) (1,2)};
\addplot[red!75!black, dashed] coordinates {(0,2) (0.5,2) (0.5,0) (1,0)};
\addplot[black!55, line width=0.6pt] coordinates {(0,0) (0.5,0) (0.5,1) (1,1)};
\addplot[black!30, densely dotted, line width=0.6pt] {x};

\nextgroupplot[title={(d) miscalibrated, inaccurate}]
\addplot[blue!70!black] coordinates {(0,0) (0.2,0) (0.2,1.25) (1,1.25)};
\addplot[red!75!black, dashed] coordinates {(0,1.25) (0.8,1.25) (0.8,0) (1,0)};
\addplot[black!55, line width=0.6pt] coordinates {(0,0) (0.2,0) (0.2,0.5) (0.8,0.5) (0.8,1) (1,1)};
\addplot[black!30, densely dotted, line width=0.6pt] {x};

\end{groupplot}
\end{tikzpicture}

\vspace{2pt}
\pgfplotslegendfromname{leg:calvsacc}   
\caption{Conditional distributions $p(\hat p|s=0)$ and $p(\hat p|s=1)$ of the predicted probability $\hat p=\hat p(s=1|o)$ given the binary state $s$, for four predictors, together with the probability $p(s=1|\hat{p})$ and with the predictive probability $\hat{p}$.}
\label{fig:cal-vs-acc}
\end{figure}
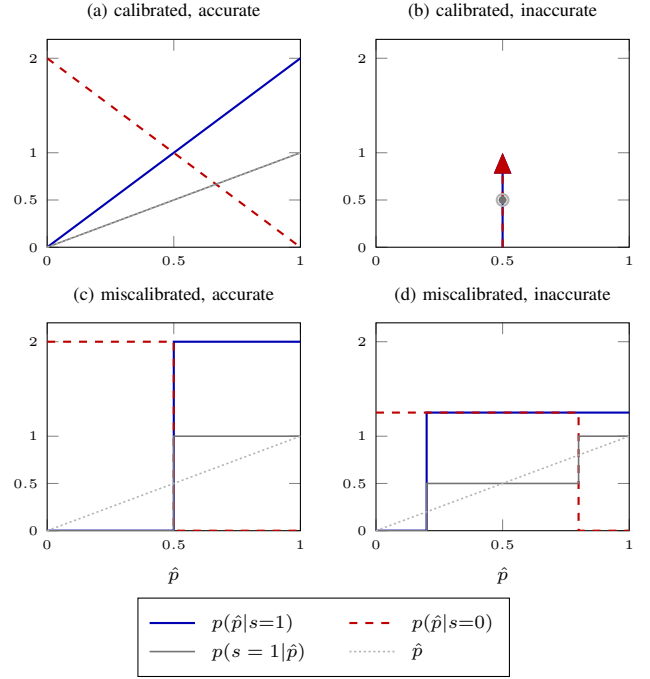}

The next proposition relates calibration to optimal decision making.

\begin{proposition}[Calibration supports optimal decisions {\cite[Prop.~1]{zhao2021calibrating}}]
\label{prop:cal}
If $\hat p(s| o)$ is distribution calibrated, then the plug-in policy that
selects $a^*(\hat p(s| o))$ in \eqref{eq:optaction1} is optimal among all
policies that depend on the observation $o$ only through the predictor, i.e.,
\begin{equation}\label{eq:optcal}
a^*(\hat p)\in\argmax_{a(\hat p)}\ \E_{o \sim p(o)}\!\left[ U\!\left(a(\hat{p})| o\right)\right], 
\end{equation} where  $\hat p$ stands for the predictor $\hat p(s| o)$.
\end{proposition}
\begin{proof}
Write $\hat p_o=\hat p(s| o)$ for the random predictive vector obtained
from $o\sim p(o)$, with distribution $p(\hat p_o)$ and joint
$p(\hat p_o,s)=p(\hat p_o)\,p(s| \hat p_o)$, where $p(s|\hat p_o)=p(s | o\in\Oset_{\hat p_o})$. The conditional distribution 
$p(s| \hat p_o)$ is thus the true state distribution across all observations
mapped to the same predictive vector $\hat p_o$. By calibration \eqref{eq:distcal}, it is given by 
$p(s| \hat p_o)=\hat p(s|o)$. For any policy $a(\hat p_o)$, conditioning on
$\hat p_o$ and using calibration, we have 
\begin{align}
\E_{(s,o)\sim p(s,o)}[u(s,a(\hat p_o))]
&=\E_{\hat p_o \sim p(\hat p_o)}\!\Bigl[\textstyle\E_{s\sim p(s|\hat p_o)}[u(s,a(\hat p_o))]\Bigr]\nonumber\\
&=\E_{\hat p_o \sim p(\hat p_o)}\!\Bigl[\textstyle\E_{s\sim \hat p(s|o)}[u(s,a(\hat p_o))]\Bigr]\nonumber\\
&=\E_{\hat p_o \sim p(\hat p_o)}\!\Bigl[\hat U(a(\hat p_o)|o)\Bigr]\nonumber\\
&\leq \E_{\hat p_o \sim p(\hat p_o)}\!\bigl[ \hat U(a^*(\hat p_o)| o)\bigr],
\end{align}
where the inequality follows from the definition \eqref{eq:optaction1}.
\end{proof}

Unlike \eqref{eq:optaction}, the optimality condition \eqref{eq:optcal}
involves an expectation over observations $o$. The price of acting through a
predictor, even a calibrated one, is that optimality is generally guaranteed
only \emph{on average}, not per observation. As Example~\ref{ex:cal} shows,
this price can be steep: a calibrated predictor may be completely
uninformative, so its action cannot adapt to the observation $o$.

{ As demonstrated in the proof of Proposition \ref{prop:cal}, given a calibrated predictor $\hat p(s|o)$, the estimate (\ref{eq:uhat}) of the utility equals the true average utility conditioned on the information available at the decision maker, namely the model's output  $\hat p(s|o)$. This property implies that a calibrated model supports decision making with  \emph{predictable} performance, i.e., with a utility level that can be directly evaluated from the available model $\hat p(s|o)$. Predictability is viewed as a key requirement for reliable decision-making systems \cite{rabanser2026towards}. }

In practice, given a predictor, the level of miscalibration must be estimated from a finite validation set.
Checking full distribution calibration \eqref{eq:distcal} would require
binning the space of predictive distributions, which is impractical beyond a
few dimensions. The most common metric, namely the \emph{expected
calibration error} (ECE) \cite{guo2017calibration},  therefore restricts attention to the
predicted probability of the most likely class. While intuitive, the ECE
is a biased and inconsistent estimator: its value depends on the arbitrary
number and width of the bins, and no binning scheme can simultaneously
average away sampling noise and retain the resolution needed to localize
miscalibration \cite{arrietaibarra2022metrics}. Binning-free alternatives
instead plot the \emph{cumulative} difference between predicted and observed
outcomes, ordered by predicted probability, and come with rigorous $p$-values
for the null hypothesis of perfect calibration, giving a consistent and
reproducible assessment \cite{arrietaibarra2022metrics}. 

When a predictor is
found miscalibrated, calibration can be partially restored by post-hoc
methods such as temperature or Platt scaling, which fit a few parameters on
held-out data to align predicted probabilities with observed frequencies
\cite{guo2017calibration,zhao2021calibrating}.

\subsection{Non-Parametric Data-Driven Predictors: Credal Sets and Robust Optimization}
\label{sec:npar}
Consider now the case with no predictor given but a dataset \eqref{eq:dataset}
available. In this section, we take both variables  $s$ and $o$ discrete, with finite cardinalities
$|\Sset|$ and $|\Oset|$. A natural first step is to evaluate the \emph{histogram}, the
simplest non-parametric estimator of $p(s| o)$. Define the sub-dataset
\begin{equation}
\label{eq:subdataset}
\Dset_o=\{(o_i,s_i)\in\Dset : o_i=o\}
\end{equation}
of the $n_o=|\Dset_o|$ pairs with observation $o$. The histogram is the
vector of state frequencies
\begin{equation}
\label{eq:emppost}
\hat p^{\text{hist}}(s| o)=\frac{1}{n_o}\sum_{i:\,o_i=o}\mathbf 1\{s_i=s\},
\end{equation}
where $\mathbf 1\{\cdot\}$ is the indicator function. The corresponding
estimated utility \eqref{eq:uhat} for a fixed action $a$ is the empirical mean
\begin{equation}
\label{eq:uhathist}
\hat U(a|\Dset_o) = \frac{1}{n_o} \sum_{(o_i,s_i) \in \Dset_o} u(s_i,a).
\end{equation} { Note that this estimate does not require any specific assumption on the environment, and is thus free of of any bias that may affect the ultimate large-sample performance of data-driven estimates \cite{simeone2022ml}.}

However, this quantity is not a reliable performance certificate, and thus it does not lead to predictable performance guarantees. To see why, conditioned on the dataset size $n_o$, consider
the \emph{out-of-sample disappointment} \cite{bental2013robust}
\begin{equation}
\label{eq:disappointment}
\Pr_{\Dset_o\sim p(s| o)^{\otimes n_o}}\!\bigl[\,U(a| o) < \hat U(a|\Dset_o)\,\bigr],
\end{equation}
i.e., the probability that the true utility falls short of its data-driven
estimate over draws of the training data. The notation $p(s| o)^{\otimes n_o}$ represents 
the product distribution of the i.i.d.\ dataset $\Dset_o$. By the central
limit theorem, this probability tends to $1/2$ as $n_o$ grows, implying that the plug-in
estimate is optimistic about half the time. Quantifying the uncertainty of
$\hat p(s| o)$ is therefore essential.

A principled remedy uses \emph{credal sets} and \emph{distributionally robust
optimization} (DRO). Define a set of conditional distributions in a
neighborhood of the empirical distribution as
\begin{equation}
\label{eq:klball}
\Pset(\Dset_o)=\bigl\{\,q(s):\ \mathrm D\!\left(q(s),\,\hat p^{\text{hist}}(s| o)\right)\le r\,\bigr\},
\end{equation}
where $\mathrm D(\cdot,\cdot)$ is a divergence between distributions and $r>0$
controls the set size. When the divergence measure  $\mathrm D$ is convex in its first argument -- as
for every $f$-divergence, including the Kullback--Leibler (KL) divergence and
total variation (TV) distance -- the set \eqref{eq:klball} is a \emph{credal
set}: a closed, convex set of probability distributions \cite{walley1991}.
The agent then optimizes the \emph{worst-case utility} over the set, solving
the DRO problem
\begin{equation}\label{eq:mmaction}
a_{\mathrm{mm}}(\Dset_o)=\arg\underbrace{\max_{a\in \mathcal{A}}\ \min_{q(s)\in \Pset(\Dset_o)} \E_{s\sim q(s)}[u(s,a)]}_{U_{\mathrm{mm}}(\Dset_o)}.
\end{equation} { When can the estimated utility $U_{\mathrm{mm}}$ be used as a reliable prediction of the performance of action $a_{\mathrm{mm}}(\Dset_o)$? Does action $a_{\mathrm{mm}}(\Dset_o)$ have any optimality property?}

Suppose the divergence and radius are chosen so that the set covers the true
posterior with probability at least $1-\delta$,
\begin{equation}
\label{eq:coveragedro}
\Pr_{\Dset_o\sim p(s| o)^{\otimes n_o}}\!\bigl[\,p(s| o)\in\Pset(\Dset_o)\,\bigr]\ge 1-\delta.
\end{equation}
For example, with the total variation (TV)  distance
$\TV(q(s),q'(s))=\tfrac12\sum_{s\in\Sset}|q(s)-q'(s)|$ and radius
\begin{equation}
\label{eq:radius}
r=\sqrt{\frac{|\Sset|\log 2+\log(1/\delta)}{2\,n_o}},
\end{equation}
a finite-sample deviation bound for the empirical distribution
\cite{weissman2003l1} guarantees
$\TV(\hat p^{\text{hist}}(s| o),\,p(s| o))\le r $ with probability at
least $1-\delta$, so that the condition  \eqref{eq:coveragedro} holds.

\begin{proposition}[Coverage implies controlled disappointment]
\label{prop:cov}
If the credal set $\Pset(\Dset_o)$ satisfies the coverage condition
\eqref{eq:coveragedro}, then the certificate $U_{\mathrm{mm}}(\Dset_o)$ in
\eqref{eq:mmaction} controls the out-of-sample disappointment of the max--min
action, in the sense that we have the inequality 
\begin{equation}
\label{eq:disappointmentbound}
\Pr_{\Dset_o\sim p(s| o)^{\otimes n_o}}\!\bigl[\,U(a_{\mathrm{mm}}(\Dset_o)| o) < U_{\mathrm{mm}}(\Dset_o) \,\bigr] \leq \delta.
\end{equation}
\end{proposition}
\begin{proof}
Let $E=\{p(s| o)\in\Pset(\Dset_o)\}$ be the event that the true posterior is included in the credal set, which by \eqref{eq:coveragedro}
satisfies $\Pr_{\Dset_o\sim p(s| o)^{\otimes n_o}}[E]\ge 1-\delta$. On the event $E$, the true posterior is feasible in the
inner minimization of \eqref{eq:mmaction}, yielding 
\begin{align}
U_{\mathrm{mm}}(\Dset_o) &=\min_{q(s)\in\Pset(\Dset_o)}\E_{s\sim q(s)}[u(s,a_{\mathrm{mm}})]
\nonumber \\&\le \E_{s\sim p(s|o)}[u(s,a_{\mathrm{mm}})]=U(a_{\mathrm{mm}}| o).
\end{align}
Hence disappointment cannot occur on event $E$, i.e., 
$\{U(a_{\mathrm{mm}}| o)<U_{\mathrm{mm}}(\Dset_o)\}\subseteq E^{c}$, and
therefore we have 
$\Pr_{\Dset_o\sim p(s| o)^{\otimes n_o}}[U(a_{\mathrm{mm}}| o)<U_{\mathrm{mm}}(\Dset_o)]\le\Pr_{\Dset_o\sim p(s| o)^{\otimes n_o}}[E^{c}]\le\delta$.
\end{proof}

The certificate $U_{\mathrm{mm}}(\Dset_o)$ is thus a reliable high-confidence
lower bound on the utility the data-driven action will actually achieve. With
probability at least $1-\delta$, the realized utility is no worse than the
certified value, so the agent may safely plan against $U_{\mathrm{mm}}(\Dset_o)$
rather than the optimistic empirical mean \eqref{eq:uhathist}. This DRO-based approach corresponds to 
the interface illustrated in  Fig.~\ref{fig:racinterface}(d).

Credal sets
satisfying \eqref{eq:coveragedro} trade off conservatism against tightness: a
larger set gives stronger coverage but a smaller, less useful certificate. A principled optimization criterion would hence ask  for the \emph{least conservative} credal set, which yields the pointwise largest
certificate among all rules with controlled disappointment. Reference \cite{vanparys2021optimal} shows that, in the
large-sample regime, choosing the metric $\mathrm D$ to be the KL divergence is optimal
in this sense, yielding actions whose certificate exhibits disappointment
decaying exponentially, as $\exp(-n_o r)$, where $r$ is the radius in (\ref{eq:klball}). 

\subsection{Parametric Data-Driven Predictors: Bayesian Learning and Inference}
\label{sec:par}
Under parametric assumptions, the state--observation relationship is modeled
by a family of conditional distributions $p(s| o,\theta)$ indexed by parameters $\theta\in\Theta$. { The family $p(s| o,\theta)$ constitutes a model class, and its selection inherently biases any  decision maker that relies on it.} Taking a Bayesian perspective, the decision maker selects a prior distribution $p(\theta)$ on the parameters, and thus on the models within the family $\{p(s| o,\theta) \}_{\theta \in \Theta}$. Epistemic uncertainty now resides in the
lack of knowledge about parameters $\theta$.

Given the dataset \eqref{eq:dataset}, Bayes'
rule yields the parameter posterior
\begin{equation}
\label{eq:parampost}
p(\theta|\Dset)\propto p(\theta)\prod_{i=1}^{n} p(s_i| o_i,\theta),
\end{equation}
and the \emph{posterior predictive} distribution of the state is obtained by
marginalizing over it, i.e., 
\begin{equation}
\label{eq:predictive}
p(s| o,\Dset)=\E_{\theta \sim p(\theta|\Dset)}\!\bigl[p(s| o,\theta)\bigr].
\end{equation}
Repeating the risk-neutral argument of Sec.~\ref{sec:known} with the conditional distribution $p(s| o)$
replaced by the predictive \eqref{eq:predictive}, the optimal action is
\begin{equation}
\label{eq:bayesact}
\begin{aligned}
\hat a(o)&=\argmax_{a \in \mathcal{A}}\ \E_{s \sim p(s| o,\Dset)}[u(s,a)]\\
&=\argmax_{a\in \mathcal{A}}\ \E_{\theta \sim p(\theta|\Dset)}\E_{s\sim p(s| o,\theta)}[u(s,a)].
\end{aligned}
\end{equation}

Hence, optimal risk-neutral decision making requires the parameter posterior
$p(\theta|\Dset)$, and the
predictive posterior \eqref{eq:predictive} plays exactly the role of the
known-environment interface $o\mapsto p(s| o)$, as illustrated in Fig.~\ref{fig:racinterface}(e). In this regard, it is worth noting that a plug-in estimate
that replaces the posterior $p(\theta|\Dset)$ by a point estimate such as maximum-likelihood, and acts on the distribution $p(s| o,\hat\theta)$
would discard this epistemic uncertainty and is again prone to out-of-sample
disappointment. In contrast, assuming the validity of the assumed model, consisting of family $\{p(s|o,\theta\}_{\theta \in \Theta}$ and prior $p(\theta)$, the average utility $\E_{s \sim p(s| o,\Dset)}[u(s,a)]$ offers a reliable estimate of the agent's performance. 

Since the normalizing constant in \eqref{eq:parampost} is
rarely tractable, one in practice resorts to approximate inference using methods such as Laplace
approximation, variational inference, or Markov chain Monte Carlo \cite{simeone2022ml}.

{

\section{Beyond Predictability}

As summarized in \cite{rabanser2026towards}, beyond predictability, which was covered in the previous sections, a reliable agent should ideally also address the requirements of consistency, robustness, and safety. This section outlines how these requirements can be incorporated within the decision-theoretic framework adopted in these notes.

\subsection{Consistency}
Consistency evaluates the variability of the outputs produced by a system when run multiple times under identical conditions. In the context of these notes, it can be practically measured by evaluating the variability of the action $a$ under the given policy $\pi(a|o)$. A class of consistency metrics can be accordingly derived by drawing independent samples $a\sim \pi(a|o)$ and $a'\sim \pi(a|o)$ under the given observation $o$ and evaluating the average distance $\mathbb{E}_{a,a'\stackrel{\text{i.i.d.}}{\sim} \pi(a|o)}[d(a,a')]$ for some distance measure $d(a,a')$ such as the quadratic distance $d(a,a')=\|a-a'\|^2$, yielding twice the trace of the covariance of the action $a\sim \pi(a|o)$, or the normalized Levenshtein distance or a semantics-aware edit distance for sequences \cite{rabanser2026towards}. Such metrics can be estimated via U-statistics.

It is generally preferable to keep this variability measure low, ensuring a consistent behavior of the agent under given conditions. This can be ensured for free, i.e., without affecting the achievable utility levels, whenever the set of optimal policies $\pi(a|o)$ includes a deterministic policy as in (\ref{eq:optaction}). This is the case in all settings studied in these notes except for the scenario studied in Sec. \ref{sec:npar}, which involves the DRO problem (\ref{eq:mmaction}). In fact, the robust policy optimization problem
\begin{equation} \max_{\pi(a|o)} \min_{q(s)\in \Pset(\Dset_o)} \E_{s\sim q(s),a\sim \pi(a|o)}[u(s,a)]\end{equation}
can be shown to be concave, but not linear, in the policy $\pi(a|o)$, so that its optimal solution may lie in the interior of the simplex, i.e., it may be randomized. Accordingly, robust optimization can benefit from randomization, causing utility maximization to conflict with consistency.

\subsection{Robustness}
Robustness refers to the performance degradation caused by deviations from nominal environmental conditions. These deviations can be captured in our setting via two qualitatively different situations:
\begin{itemize}
\item \emph{Task deviation:} The state distribution $p(s)$ may deviate from the nominal distribution $p_0(s)$. When the magnitude of this deviation can be quantified, such a deviation may be accounted for by assuming that the true state distribution belongs to a credal set of the form $\{p(s):\mathrm{D}(p(s)\|p_0(s))\leq r_S\}$ for some divergence measure $\mathrm{D}(\cdot\|\cdot)$ and radius $r_S\geq 0$.
\item \emph{Evidence deviation:} The observation model $p(o|s)$ may deviate from a nominal model $p_0(o|s)$, belonging to a set of the form $\{p(o|s):\mathrm{D}(p(o|s)\|p_0(o|s))\leq r_O\}$ for some radius $r_O \geq 0$ and all states $s\in\mathcal{S}$.
\end{itemize}
Utility optimization under robustness requirements can thus be addressed via DRO formulations that account for the uncertainty on the true environment captured by the credal sets describing task or evidence deviations. Accordingly, following the discussion above, maximizing utility under robustness is generally in tension with the goal of implementing consistent policies.

\subsection{Safety}
Following \cite{rabanser2026towards}, safety requirements address the question: When failures occur, how severe are the resulting consequences? To study this question in our setting, we focus here on the risk-averse agent considered in Sec. \ref{sec:ra}, as this setting provides the most natural framework to evaluate the issue of safety as it relates to utility optimization.

In this setting, fix a deterministic policy yielding an action $a(o)$ and the corresponding utility $u(s,a(o))$, which is assumed for simplicity to have a continuous distribution under the random state  $s\sim p(s|o)$. Consider also the utility certificate $V_{\alpha}(a(o)|o)$ in \eqref{eq:var}, i.e., the value-at-risk at level $\alpha$. A \emph{failure} occurs for states $s$ that yield the inequality $u(s,a(o))<V_{\alpha}(a(o)|o)$, an event denoted as
\begin{equation}
  \mathcal{E}_{o}=\bigl\{s\in\mathcal{S}:\ u(s,a(o))<V_{\alpha}(a(o)|o)\bigr\},
  \label{eq:failure_event}
\end{equation}
which occurs with probability $\alpha$, i.e., $\Pr_{s\sim p(s|o)}[\mathcal{E}_{o}]=\alpha$ by the definition \eqref{eq:var}.

The certificate provided by the value-at-risk only constrains the \emph{frequency} of failure, while providing no information about its \emph{severity}. A safety objective is thus the \emph{conditional value-at-risk}
\begin{equation}
  \CVaR_{\alpha}(o)
  =\mathbb{E}_{s\sim p(s|o)}\!\left[\,u(s,a(o))\,\middle|\,\mathcal{E}_{o}\,\right],
  \label{eq:cvar_cond}
\end{equation}
i.e., the average utility over the $\alpha$-fraction of worst-case states. This criterion can be expressed as the DRO problem \cite{ang2018dual}
\begin{equation}
  \CVaR_{\alpha}(o)
  =\min_{q(s)\in\mathcal{P}(o)}\ \mathbb{E}_{s\sim q(s)}\!\left[u(s,a(o))\right],
  \label{eq:cvar_dro}
\end{equation}
over the credal set of distributions
\begin{equation}
  \mathcal{P}(o)=\left\{q(s):
  \Dinf\!\left(q(s)\,\|\,p(s|o)\right)\leq-\log_{2}(\alpha)\right\},
  \label{eq:credal_cvar}
\end{equation}
with the max-relative entropy, i.e., the R\'enyi divergence of order infinity \cite{Simeone2026CQIT}
\begin{equation}
  \Dinf\!\left(q(s)\,\|\,p(s|o)\right)
  =\sup_{s\in\mathcal{S}}\ \log_{2}\!\left(\frac{q(s)}{p(s|o)}\right).
  \label{eq:dmax}
\end{equation}
The constraint in \eqref{eq:credal_cvar} is equivalent to the pointwise bound $q(s)\leq p(s|o)/\alpha$, and the minimizer of \eqref{eq:cvar_dro} is the posterior restricted to the failure event, i.e., $q^{\star}(s)=\frac{1}{\alpha}\,p(s|o)\,\mathrm{1}\{s\in\mathcal{E}_{o}\}$, recovering the original definition \eqref{eq:cvar_cond}.

Accordingly, controlling the safety measure given by the CVaR reduces to the problem of addressing a DRO over a credal set whose size increases as the outage probability $\alpha$ decreases. Since enlarging the credal set can only decrease the minimum in \eqref{eq:cvar_dro}, the CVaR is non-decreasing in $\alpha$ and satisfies $\CVaR_{\alpha}(o)\leq V_{\alpha}(a(o)|o)$, tending to $\inf_{s\in\mathcal{S}}u(s,a(o))$ as $\alpha\to0$: smaller values of $\alpha$ yield smaller average utilities, evaluated over a smaller and more extreme set of worst-case states. The gap $V_{\alpha}(a(o)|o)-\CVaR_{\alpha}(o)$ measures how much worse than its own certificate the agent can expect to do when that certificate fails.


}

\section{What We Have Learned}
{ The recurring lesson in these notes is that a reliable decision needs an uncertainty
representation {matched to its objective}, and a {guarantee} that
certifies the value actually obtained, offering a predictable performance level \cite{rabanser2026towards}.}

\begin{itemize}
\item In a \emph{known} environment, a risk-neutral agent needs only the
posterior $p(s| o)$ and maximizes expected utility, while a risk-averse agent
needs a \emph{prediction set} and maximizes worst-case utility. Coverage of
the set certifies a value-at-risk (Lemma~\ref{lem:guarantee}), and a suitably designed set comes with no loss of optimality  (Proposition~\ref{prop:ra}).

\item With a \emph{fixed predictor}, distribution calibration makes the
plug-in action optimal, but only \emph{on average} over observations
(Proposition~\ref{prop:cal}). Calibration is neither accuracy nor
informativeness (Example~\ref{ex:cal}), but a calibrated model yields reliable performance guarantees. 

\item With \emph{finite data}, the empirical mean is an optimistic, and hence
untrustworthy, criterion causing out-of-sample disappointment. Leveraging such a non-parametric predictor, a credal set that
covers the truth with high probability turns the max--min value into a
high-confidence certificate (Proposition~\ref{prop:cov}), and the KL-based rule is
asymptotically the least conservative such choice.

\item Under \emph{parametric} models, the parameter posterior is the
sufficient statistic, and the posterior predictive plays the role of the
known-environment interface. Point
estimates discard this uncertainty and reintroduce disappointment.
\end{itemize}

Calibration, conformal prediction, credal sets and DRO, as well as Bayesian inference
are usually presented as separate topics. Seen through a single decision
problem, they are complementary answers to one question, i.e., \emph{what is a
sufficient, trustworthy summary of uncertainty for the decision at
hand?} Together, they demonstrate that reliable decision making \cite{rabanser2026towards} needs
uncertainty quantification. { Furthermore, uncertainty quantification must be also tailored to requirements beyond predictability, including consistency, robustness, and safety.}

\section*{Author}
\emph{Osvaldo Simeone} (o.simeone@nulondon.ac.uk) is the Professor of Information Engineering at Northeastern University London, where he co-directs the Institute for  Intelligent Networked Systems (INSI), and a visiting Professor at Aalborg University. Prof. Simeone is the author of the textbooks "Machine Learning for Engineers" and "Classical and Quantum Information Theory" published by Cambridge University Press, four monographs, two edited books, and more than 250 research journal and magazine papers. He is a Fellow of the IET and IEEE.

\bibliographystyle{IEEEtran}
\bibliography{refs}

\end{document}